\documentclass{article}
\usepackage[margin=2.4cm]{geometry}

\PassOptionsToPackage{numbers, compress}{natbib}

\PassOptionsToPackage{prologue,dvipsnames}{xcolor}

\usepackage[utf8]{inputenc} 
\usepackage[T1]{fontenc}    
\usepackage[hidelinks]{hyperref}       
\usepackage{url}            
\usepackage{booktabs}       
\usepackage{amsfonts}       
\usepackage{nicefrac}       
\usepackage{microtype}      

\let\oldemph\emph
\renewcommand{\emph}[1]{{\color{BrickRed}\oldemph{#1}}}

\usepackage{wrapfig}

\usepackage{amsmath}
\usepackage{amssymb}
\usepackage{mathtools}
\usepackage{amsthm}
\usepackage{bbm}

\usepackage[capitalize,noabbrev,nameinlink]{cleveref}

\theoremstyle{definition}
\newtheorem{theorem}{Theorem}

\newtheorem{corollary}[theorem]{Corollary}

\newtheorem{definition}{Definition}

\theoremstyle{remark}

\usepackage[textsize=tiny]{todonotes}

\usepackage{booktabs}       
\usepackage{nicefrac}       
\usepackage{microtype}      
\usepackage{xcolor}         
\usepackage{xspace}

\usepackage{comment}
\usepackage[T1]{fontenc}

\usepackage{tikz}
\usetikzlibrary{patterns}
\usetikzlibrary{calc}
\usetikzlibrary{arrows.meta}
\usetikzlibrary{positioning}
\usetikzlibrary{shapes.geometric}
\usetikzlibrary{decorations.pathreplacing}
\usepackage{pgfplots}
\pgfplotsset{compat=1.18}

\usepackage{mathtools}

\usepackage[T1]{fontenc}

\usepackage{pifont}
\definecolor{crimson}{RGB}{220, 20, 60}
\definecolor{darkgreen}{rgb}{0,0.5,0}

\definecolor{mossgreen}{rgb}{0.68, 0.87, 0.68}

\usepackage{amsthm}
\newtheorem{example}{Example}

\usepackage{enumitem}
\setlist[description]{topsep=0pt, partopsep=0pt, itemsep=2pt, leftmargin=2em}

\usepackage{cleveref}

  \tikzset{     
    e4c node/.style={circle,draw,minimum size=0.3cm,inner sep=0,font=\tiny},
    e4b node/.style={circle,draw,minimum size=0.8cm,inner sep=0,font=\normalsize},
    selected/.style={draw=BrickRed,double}, 
    selected2/.style={preaction={fill, OliveGreen!45}}, 
    selected3/.style={}, 
    e4c edge/.style={sloped,above,font=\footnotesize,-{Classical TikZ Rightarrow[length=0.75mm]},draw,thin},
    e4c path/.style={-{Classical TikZ Rightarrow[length=0.75mm]},draw,thin}
  }

\tikzset{
  efc-node/.style={
    circle,
    draw=black,
    line width=0.6pt,
    minimum size=0.70cm,
    inner sep=0pt,
    fill=white,
    font=\fontsize{5}{5.5}\selectfont
  },
  efc-edge/.style={
    ->,
    draw=black,
    line width=0.5pt
  }
}
 
\tikzset{>=stealth}

\usepackage{natbib}
\usepackage{authblk}

\author{
Georgios Papasotiropoulos\\
{\small
University of Warsaw \vspace{-0.5cm}
}
}
  \date{}

\usepackage{aliascnt}
\usepackage{cleveref}[capitalize]

\usepackage{comment}

\Crefname{corollary}{Corollary}{Corollaries}

\title{Reverse Sequential Proportional Approval Voting Rule:\\ Proportionality and Approximation Guarantees}
\date{}

\begin{document}

\maketitle

\begin{abstract}
\noindent We study the Reverse Sequential Proportional Approval Voting Rule (RevSeqPAV) in approval-based committee elections. Despite its historical prominence and practical use, its properties and guarantees are much less understood than those of Sequential PAV. We analyze it along two dimensions: proportional representation (measured by Extended Justified Representation, its approximations, and proportionality degree) and approximation of the maximum PAV score of instances. We first establish strong negative results for general, unrestricted election instances and then identify settings in which the rule provides meaningful fairness and optimization guarantees.
\end{abstract}

\section{Introduction}
\label{sec:intro}
More than a century before approval-based committee voting became a standard topic in computational social choice \cite{lac-sko:multiwinner-book}, Thorvald N.~Thiele had already proposed some of its central rules \cite{Thie95a}. His main method for electing $k$ out of $m$ candidates based on voters' approval ballots assigns to each committee the sum, over all voters, of the harmonic score $1+\frac12+\cdots+\frac1t$, where $t$ is the number of committee members approved by the voter. This quantity is now known as the PAV score, and a committee maximizing it is selected. The resulting voting rule is \emph{Proportional Approval Voting} (PAV).
Having observed the combinatorial blow-up in the number of possible committees as the number of candidates grows, Thiele also proposed two greedy methods for approximating the PAV objective: one that builds the committee candidate by candidate, each time selecting the candidate with the largest marginal contribution to the PAV score of the current solution, and one that starts from all candidates and iteratively eliminates the candidate whose removal causes the smallest decrease in the current PAV score until reaching a committee of size $k$. In modern terminology, these are \emph{Sequential PAV} (SeqPAV) and \emph{Reverse Sequential PAV} (RevSeqPAV), respectively.

The harmonic function that appears in PAV scores discourages repeatedly representing the same voters and leads to strong proportionality guarantees, most notably Extended Justified Representation (EJR) \cite{aziz2017justified}. However, finding a PAV-optimal committee is NP-hard \cite{azi-gas-gud-mac-mat-wal:c:multiwinner-approval,owaWinner}, making its two simple polynomial-time greedy variants natural alternatives for optimizing the objective of PAV \cite{lac-sko:multiwinner-book}.
Despite their close connection, the two rules have received very different levels of attention in the computational social choice literature. SeqPAV has been studied extensively; see, indicatively, \cite{boehmer2024approval,faliszewski2022robustness,faliszewski2024experimental,israel2025dynamic,janeczko2023ties,malysociotropic,sanchez2024maximin}, whereas its reverse counterpart remains much less understood. This is somewhat surprising, given that, of the two methods, Thiele himself recommended the elimination-based one \cite{Janson16arxiv}.

What is known about RevSeqPAV gives a mixed picture. It satisfies important axiomatic properties like committee monotonicity (in contrast to PAV), support monotonicity (in contrast to SeqPAV) and D'Hondt proportionality \cite{lac-sko:multiwinner-book}, and, while it fails even the mild fairness guarantee of Justified Representation (JR), it satisfies it whenever the number of candidates exceeds the desired committee size by at most two \cite{aziz2017note}. \citet{lac-sko:multiwinner-book} describe RevSeqPAV as a ``non-standard'' approval-based committee rule because, unlike most such rules, it need not satisfy the natural requirement of electing the most approved candidate when only one candidate is to be selected.
In an extensive empirical comparison of approval-based committee rules, \citet{lackner2020utilitarian} found PAV, SeqPAV, and RevSeqPAV among the strongest compromises between utilitarian welfare and representation, and remarked that the latter two are ``virtually indistinguishable'' in their experiments. Similar findings were reported by \citet{proprank}, who studied RevSeqPAV, among others, in the context of producing a ranking over the candidates. 
\looseness-1

Despite the limited theoretical attention it has received, RevSeqPAV has seen real-world use.
\emph{LiquidFeedback}, a platform for online democratic deliberation and collective decision making, uses RevSeqPAV (under the name ``Harmonic Weighting'') for sorting participants' proposals in a way that fairly reflects the electorate's preferences \cite{behrens2014evolution,behrens2014principles}. The formal fairness guarantees of this mechanism have generated some history: a stronger proportionality claim was initially made to justify its use in LiquidFeedback, but was later shown to be false by \citet{aziz2017note}, leading to a relevant corrigendum \cite{liquidfeedback2018corrigendum}. 
Understanding and formalizing the rule's potential and limitations is not only of theoretical interest but also practically relevant.

\subsection{Contributions}
This work provides an analysis of RevSeqPAV along two dimensions: its guarantees for \emph{proportional representation} of the electorate's preferences and its ability to \emph{approximate the maximization objective of PAV}.

On the proportionality side, we  resolve the open question from the textbook by \citet{lac-sko:multiwinner-book} concerning the proportionality degree of RevSeqPAV (see Q9 in Sec. 7.2 and note that RevSeqPAV is the only rule among those considered for which bounds on this fairness metric are missing from Table 4.1 of the book). The proportionality degree of a rule corresponds to a guarantee on the average satisfaction that its outcome provides to voters within a group that, according to the premise of EJR, deserves a satisfaction of $\ell$ \cite{skowron:prop-degree}.
We prove that RevSeqPAV has proportionality degree zero at every level $\ell$: for each $\ell$, there is an election containing an $\ell$-cohesive group whose members obtain no approved winner. This is a strongly negative result for the rule's worst-case fairness guarantees, standing in sharp contrast to those of PAV 
 $(=\ell-1)$ and SeqPAV $(\approx0.7\ell-1,$ for $k\leq 200$).
    Notably, the same construction also settles a question from the proportional-rankings literature explicitly left open by \citet{proprank}. For the measure of proportional representation introduced in that work, our construction shows that RevSeqPAV has a worst-case guarantee of zero.
    
The picture changes sharply when only a very small number of candidates have to be eliminated; we denote this quantity by $r:=m-k$. This was already known to some extent from \citet{aziz2017note}, which establishes JR when at most two candidates ($r\leq2$) have to be eliminated. We strengthen this axiomatic guarantee from JR to EJR. 
We view this result primarily as an existential one, identifying a class of elections on which RevSeqPAV always satisfies EJR, whereas SeqPAV may fail even JR for both $r=1$ and $r=2$.
Nevertheless, there are real-life scenarios in which the desired committee size is close to the number of candidates, for instance when selecting a board for a small organization, or when the candidates have already undergone a shortlisting process. 
Additionally, we show that the requirement of $r\leq 2$ is tight: for every $r\geq 3$ we provide an instance where the outcome of RevSeqPAV violates EJR.

Despite the negative result for EJR, we show that RevSeqPAV always guarantees an approximation of this notion. In particular, the rule satisfies $(1/r)$-EJR, which further formalizes the intuition that the rule performs well when $k$ is close to $m$. In the course of proving this guarantee, we also identify a bound on the size of groups with common preferences that suffices to ensure their representation in every instance; thus, even when EJR fails, sufficiently large groups still receive the level of representation prescribed by the axiom.

On the optimization side, we first answer another relevant question posed by \citet{lac-sko:multiwinner-book} (Q17 in Sec.~7.2), which asks whether RevSeqPAV can achieve an approximation factor for the optimal PAV score that is better than the $1-\frac{1}{e}$ that SeqPAV achieves. Our answer is firmly negative: RevSeqPAV may return a committee whose PAV score is an arbitrarily small fraction of the optimum. 

As in our proportionality analysis, we complement this worst-case result with two positive regimes. First, when the goal is to select a committee of size $k$ that constitutes a large fraction of the $m$ candidates, we show that RevSeqPAV remains close to the optimal PAV score. Our bound of $\frac{k}{m-1}$ converges to $1$ as $k$ tends to infinity for every fixed number of deletions $r$, providing a quantitative version of the intuition that this reverse elimination method is particularly attractive when $k$ is close to $m$ \cite{aziz2017note}. Second, under the assumption that every voter approves at most $b$ candidates, we prove that RevSeqPAV admits an approximation guarantee of $\frac1b$. We also show tightness of this bound.

Taken together, our results provide a systematic account of the potential and limitations of this century-old elimination rule, identifying its worst-case shortcomings together with settings in which it provides strong guarantees.\looseness-1

\section{Preliminaries}

An \emph{election} is a tuple $E=(N,C,A,k)$, where $N=\{1,\ldots,n\}$ is a set of voters, $C$ is a set of $m$ candidates, $A=(A_i)_{i\in N}$ is a profile of approval ballots with $A_i\subseteq C$ denoting the set of candidates approved by voter $i$, and $k$ is the desired committee size. 
We assume that $1\leq k<m$ and that $A_i\neq \emptyset$ for at least one voter $i$.
A (feasible) \emph{committee} is a set $W\subseteq C$ with $|W|=k$.
A \emph{voting rule} $\mathcal{R}$ is a function that given an election $E$ returns a feasible committee $\mathcal{R}(E)$.
We denote by $r$ the \emph{number of candidates that have to be eliminated} from $C$ in order to obtain a feasible committee, that is, $r=m-k$.
For a voter $i$ and a candidate set $S$, we refer to $|A_i\cap S|$ as the \emph{satisfaction} of voter $i$ from $S$, denoted by  $\operatorname{sat}_i(S)$. For a non-empty group of voters $V\subseteq N$, its \emph{average satisfaction} from $W$ is
$\operatorname{avg}_V(W)=\frac{1}{|V|}\sum_{i\in V}\operatorname{sat}_i(W)$. 
We also use
$b=\max_{i\in N}|A_i|$
for the \emph{maximum ballot size}.

\begin{definition}
Given an election $E=(N,C,A,k)$, the PAV score of a set of candidates $S$ is
$\operatorname{PAV}(S)=\sum_{i\in N}H(\operatorname{sat}_i(S))$, where, for a nonnegative integer $t$, $H(t)$ denotes the $t$-th harmonic number; $H(0)=0$ and
$H(t)=1+\frac12+\cdots+\frac1t$, for $t\ge1$.
PAV is the voting rule that returns a feasible committee maximizing PAV score; we denote the maximum PAV score under $E$ by $\operatorname{OPT}_k$.
\end{definition}

For a set $S\subseteq C$ and a candidate $c\in S$, we denote the \emph{marginal loss} from deleting $c$ from $S$ by
$\Delta_S(c)=\operatorname{PAV}(S)-\operatorname{PAV}(S\setminus\{c\})$.
This quantity has a simple voter-wise interpretation: if voter $i$ currently approves $t\geq 1$ candidates in $S$, then this voter contributes $\frac{1}{t}$ to $\Delta_S(c)$ when $c\in A_i$ and $0$ otherwise.
We use $\Delta(c)$ when the set $S$ is clear from the context.
RevSeqPAV iteratively removes the candidate whose deletion causes the smallest immediate decrease in the PAV score, i.e., the one of smallest marginal PAV loss. We resolve ties using a fixed priority order on candidates.

\begin{definition}
Let $C_m=C$. For every $s=m,m-1,\ldots,k+1$, RevSeqPAV identifies the candidate
$c_s\in\arg\min_{c\in C_s}\Delta_{C_s}(c)$
and sets
$C_{s-1}=C_s\setminus\{c_s\}$.
The rule returns the committee $W=C_k$.
\end{definition}

\begin{example}
Consider an election with $C=\{a,b,c\}$ and $k=2$. Three voters approve $\{a,b\}$, three other voters approve $\{a,c\}$, while two voters approve only $\{b\}$ and another two approve only $\{c\}$.
RevSeqPAV starts from $C_3=\{a,b,c\}$. The initial marginal losses are
$\Delta_{C_3}(a)=3\cdot\frac12+3\cdot\frac12=3$,
$\Delta_{C_3}(b)=3\cdot\frac12+2=\frac72 = \Delta_{C_3}(c)$.
Hence, RevSeqPAV deletes $a$ and returns $\{b,c\}$, which in this example coincides with the PAV-optimal committee, while SeqPAV would first select $a$ and return either $\{a,b\}$ or $\{a,c\}$, depending on tie-breaking.
\end{example}

We introduce below the two concepts of proportional representation on which \Cref{sec:prop} focuses. Both provide guarantees to sufficiently large groups of voters with common preferences, with the first (EJR) considering the maximum satisfaction attained by a voter within the group and the second (proportionality degree) their average satisfaction. We also define two relaxations of EJR which are due to \citet{aziz2017justified} and to \citet{skowron:prop-degree} and \citet{do2022online}.

\begin{definition}
\label{def:ejr}
For an integer $\ell\ge1$, a group of voters $V\subseteq N$ is \emph{$\ell$-cohesive} if
$|V|\ge \ell n/k$
and
$|\bigcap_{i\in V}A_i|\ge\ell$.
A committee $W$ satisfies \emph{Extended Justified Representation} (EJR) if, for every $\ell\ge1$ and every $\ell$-cohesive group $V$, there exists a voter $i\in V$ such that
$|A_i\cap W|\ge\ell$.
If the guarantee is satisfied for $\ell=1,$ then $W$ satisfies \emph{Justified Representation} (JR).
If for every integer $\ell\ge1$ and every group $V\subseteq N$ such that
$|V|\ge \frac{1}{\alpha}\ell\frac{n}{k}$
and
$|\bigcap_{i\in V}A_i|\ge\ell$,
there exists a voter $i\in V$ with
$|A_i\cap W|\ge\ell$ then we say that $W$ satisfies $\alpha$-EJR, $\alpha \in (0,1]$.
A rule $\mathcal{R}$ satisfies EJR (respectively, JR and $\alpha$-EJR) if for every election $E,$ it holds that $\mathcal{R}(E)$ satisfies EJR (respectively, JR and $\alpha$-EJR).
\end{definition}

\begin{definition}
A rule $\mathcal{R}$ has \emph{proportionality degree} at least $d(\ell)$ at level $\ell$ if, for every election $E$, and for every $\ell$-cohesive group $V$,
the following holds: 
$\operatorname{avg}_V(\mathcal{R}(E))\ge d(\ell)$.
We denote by $d_{\mathcal{R}}(\ell)$ the largest value for which this guarantee holds for $\mathcal{R}$. 
\end{definition}

In order to evaluate how well RevSeqPAV approximates the objective underlying PAV, for $\alpha\in [0,1]$ we say that a rule $\mathcal{R}$ has approximation factor $\alpha$ of the optimal PAV-score if
$\operatorname{PAV}(\mathcal{R}(E))\ge\alpha\operatorname{OPT}_k,$
for every election $E$.
The worst-case approximation ratio of a rule $\mathcal R$ is
$\inf_E \frac{\mathrm{PAV}(\mathcal R(E))}{\mathrm{OPT}_k}$.

\section{Proportionality Guarantees of RevSeqPAV}
\label{sec:prop}
Recall that for $r\leq 2$, RevSeqPAV satisfies the (much weaker than EJR) proportionality axiom of JR \cite{aziz2017note}. Our first result strengthens this guarantee to EJR.

\begin{theorem}
\label{thm:revseq-ejr-two-deletions}
Let $W$ be a committee returned by Reverse Sequential PAV. If at most two candidates are deleted, i.e., if $r\le 2$, then $W$ satisfies EJR.
\end{theorem}

\begin{proof}
Suppose first that $m=k+1$. Then $W$ is a PAV-optimal committee of size $k$. Since PAV satisfies EJR, so does $W$. It remains to consider $m=k+2$. Let $a$ be the first deleted candidate and $b$ the second one. Suppose, towards a contradiction, that $W$ violates EJR. Then there exist some $\ell$ and an $\ell$-cohesive group $V$ such that every voter $i\in V$ satisfies $|A_i\cap W|\le \ell-1$. Consider a set $T$ of $\ell$ candidates that are approved by every voter in $V$. Since each voter in $V$ obtains fewer than $\ell$ approved winners, at least one member of $T$ must be deleted.
We distinguish two cases.
\begin{description}
    \item[Case 1.] Say that $b\in T$. 
    Immediately before $b$ is deleted, the current set consists of $W\cup\{b\}$. Every voter in $V$ approves $b$ and at most $\ell-1$ members of $W$. Hence every such voter approves at most $\ell$ candidates from $W\cup\{b\}$. Therefore their contribution to the marginal loss of $b$ is at least $1/\ell$, and thus
$\Delta_{W\cup\{b\}}(b)\ge |V|/\ell\ge n/k$.

Consider a candidate set $C'$. If a voter $i$ approves $t>0$ members of $C'$, then the removal of any such approved candidate $c$ contributes $1/t$ to $\Delta_{C'}(c)$. Since there are exactly $t$ such candidates, voter $i$ contributes in total $t\cdot (1/t)=1$ to $\sum_{c\in C'}\Delta_{C'}(c)$. If voter $i$ approves no member of $C'$, then their contribution is $0$. Therefore,
$\sum_{c\in C'}\Delta_{C'}(c)=|\{i\in N:A_i\cap C'\neq\emptyset\}|\le n$.

Given that $|W\cup\{b\}|=k+1$ and that $\sum_{c\in W\cup\{b\}}\Delta_{W\cup\{b\}}(c)\le n$, there exists a candidate $\hat c\in W\cup\{b\}$ such that
$\Delta_{W\cup\{b\}}(\hat c)\le n/(k+1)$;
otherwise, if every candidate had marginal loss strictly greater than $n/(k+1)$ then their total marginal loss would exceed $n$. Since Reverse Sequential PAV deletes a candidate of minimum marginal loss, its choice of $b$ implies
$\Delta_{W\cup\{b\}}(b)\le \Delta_{W\cup\{b\}}(\hat c)\le n/(k+1)$.
This contradicts $\Delta_{W\cup\{b\}}(b)\ge n/k$, since $n/k>n/(k+1)$.
\item[Case 2.] Say now that $b\notin T$.
Since at least one candidate in $T$ must be deleted and the only deleted candidates are $a$ and $b$, we have $a\in T$. Hence all $\ell-1$ candidates in $T\setminus\{a\}$ belong to $W$.
The group $V$ witnesses a violation of EJR, so every voter $i\in V$ satisfies $|A_i\cap W|\le \ell-1$. It follows that no voter in $V$ can approve any winner outside $T\setminus\{a\}$; otherwise, that voter would approve at least $\ell$ members of $W$.

Let $x$ be the number of voters in $V$ who approve $b$. Before the first deletion, every voter in $V$ approves the $\ell$ candidates in $T$, and may additionally approve $b$, but approves no other candidate. More precisely, the $|V|-x$ voters who do not approve $b$ have exactly $\ell$ approved candidates in the initial candidate set, namely the candidates in $T$, whereas the $x$ voters who approve $b$ have exactly $\ell+1$ approved candidates, namely the candidates in $T\cup\{b\}$.

Since every voter in $V$ approves $a$, the contribution of the first group to the initial marginal loss of $a$ is $(|V|-x)/\ell$, while the contribution of the second group is $x/(\ell+1)$. Hence the total contribution of voters in $V$ to the marginal loss of $a$ is
$(|V|-x)/\ell+x/(\ell+1)$.
By contrast, only the $x$ voters in the second group approve $b$, and each of them initially has $\ell+1$ approved candidates. Thus the total contribution of voters in $V$ to the initial marginal loss of $b$ is
$x/(\ell+1)$.

Let $\alpha$ and $\beta$ denote the contributions of voters outside $V$ to the initial marginal losses of $a$ and $b$, respectively. Since Reverse Sequential PAV deletes $a$ before $b$, we have
$(|V|-x)/\ell+x/(\ell+1)+\alpha\le x/(\ell+1)+\beta$.
Hence
$\beta\ge (|V|-x)/\ell+\alpha\ge (|V|-x)/\ell$.
Observe that, as Reverse Sequential PAV proceeds, the contribution of any voter to the marginal loss of a fixed candidate $c$ that she approves and that has not yet been deleted can only increase. Hence, after $a$ is deleted, the contribution of voters outside $V$ to the marginal loss of $b$ is still at least $(|V|-x)/\ell$.
Moreover, after $a$ is deleted, each of the $x$ voters in $V$ who approves $b$ has exactly $\ell$ approved remaining candidates: the $\ell-1$ members of $T\setminus\{a\}$ and $b$. Hence these voters contribute $x/\ell$ to the marginal loss of $b$. It follows that
$\Delta_{W\cup\{b\}}(b)\ge (|V|-x)/\ell+x/\ell=|V|/\ell\ge n/k$.
Therefore, again, some candidate among the $k+1$ remaining candidates must have marginal loss at most $n/(k+1)$. Thus, $b$ cannot be the candidate selected for deletion.
\end{description}
Since both cases lead to a contradiction the proof is complete.
\end{proof}

While this positive guarantee applies to the rather specific case of $m-k\leq 2$, as noted in \Cref{sec:intro}, there are reasonable real-world scenarios in which $k$ can be this close to $m$. Unfortunately, the guarantee cannot be extended further, as established by the tightness result below.

\begin{theorem}
\label{thm:revseq-ejr-three-deletions}
For every $r\ge3$ there exists an election with $m-k=r$ in which the committee returned by Reverse Sequential PAV violates EJR.
\end{theorem}

\begin{proof}
Consider an election with $k=12$, $m=15$, $n=600$ and let $C =\{c_1,c_2,x_1,x_2,y_1,\ldots,y_{11}\}$.
The approval profile is as follows.

\begin{center}
\begin{tabular}{c|c}
Number of voters & Approved candidates\\
\hline
$50$ & ${c_1,c_2,x_1}$\\
$50$ & ${c_1,c_2,x_2}$\\
$17$ & ${x_1}$\\
$17$ & ${x_2}$\\
$1$ & ${c_1}$\\
$465$ & ${y_1,\ldots,y_{11}}$
\end{tabular}
\end{center}

Let $V$ consist of the first $100$ voters. Since
$|V|=100=2n/k$
and all voters in $V$ commonly approve $c_1$ and $c_2$, the group $V$ is $2$-cohesive.

We first present at a high level how RevSeqPAV works in this instance. When the temporary committee consists of the whole set of candidates $C$, the presence of \(x_1\) and \(x_2\) in the committee makes \(c_2\) look redundant. Initially, each half of the aforementioned cohesive group has three approved candidates, so their common candidates have a small marginal loss causing RevSeqPAV to delete \(c_2\). But then it subsequently deletes \(x_1\) and \(x_2\) as well. In other words, the alternatives that made \(c_2\) look dispensable do not survive themselves and this leads to a violation of EJR witnessed by $V$.

We now follow in detail the Reverse Sequential PAV deletions. Initially,
$\Delta(c_2)=100/3$,
$\Delta(x_1)=\Delta(x_2)=50/3+17=101/3$,
$\Delta(c_1)=100/3+1=103/3$,
and
$\Delta(y_j)=465/11$
for every $j$.
Thus $c_2$ is deleted first.
After deleting $c_2$,
$\Delta(x_1)=\Delta(x_2)=50/2+17=42$,
$\Delta(c_1)=100/2+1=51$,
and
$\Delta(y_j)=465/11>42$.
Hence one of $x_1$ and $x_2$ is deleted next. Suppose it is $x_1$ using some fixed tie-breaking rule.
After deleting $x_1$, we still have
$\Delta(x_2)=42$,
while
$\Delta(y_j)=465/11>42$.
Moreover, the marginal loss of $c_1$ has increased to
$50+25+1=76$.
Therefore $x_2$ is deleted third.
The resulting committee is
$W=\{c_1,y_1,\ldots,y_{11}\}$.
Every voter in the $2$-cohesive group $V$ approves exactly one member of $W$, namely $c_1$, hence EJR is violated.

This proves the claim for $r=3$. For any $r>3$, add $r-3$ candidates approved by no voter. Each such candidate has marginal loss $0$, so all of them are deleted before any candidate in the construction above. Afterwards the same three deletions occur and the same EJR violation results. Hence EJR can fail for every $r\ge3$.\looseness-1
\end{proof}

On the positive side, while EJR cannot be guaranteed for arbitrary instances, an approximation of it is possible---refer to \Cref{def:ejr} for the relevant definition. Our result provides a first bound, from which an easier-to-interpret guarantee (\Cref{cor:approxejr}) follows directly.\looseness-1

\begin{theorem}
\label{thm:revseq-alpha-ejr}
Let $r\geq3$ and set
$\alpha=\frac{k+r}{kr}$.
For every $\varepsilon>0$ such that
$0<\alpha-\varepsilon\leq1$,
the committee returned by Reverse Sequential PAV
satisfies $(\alpha-\varepsilon)$-EJR.
\end{theorem}

\begin{proof}
Suppose, towards a contradiction, that there exist an integer
$\ell\ge1$ and a group $V$ such that
$|V|\ge \frac{\ell n}{(\alpha-\varepsilon)k}$,
the voters in $V$ commonly approve at least $\ell$ candidates,
and every voter $i\in V$ satisfies $|A_i\cap W|\le\ell-1$, where $W$ is the outcome of RevSeqPAV.
Let $T$ be a set of $\ell$ candidates commonly approved by all voters in $V$. Since no voter in $V$ obtains $\ell$ approved winners, at least one candidate from $T$ must be deleted.
Consider one such commonly approved candidate $c$ at the moment when it is deleted, and let $t$ be the number of deletions that still remain at that moment, counting the deletion of $c$ itself. Thus, the current candidate set has size $k+t$, where $1\le t\le r$.

Each voter in $V$ approves at most $\ell-1$ final winners. In addition, at most the $t$ candidates that will still be deleted can currently be approved by that voter. Hence every voter in $V$ currently approves at most
$\ell-1+t$
candidates.
Since every voter in $V$ approves $c$, their total contribution to the marginal loss of $c$ is at least
$\frac{|V|}{\ell-1+t}$.
Using $|V|\ge \frac{\ell n}{(\alpha-\varepsilon)k}$, we obtain
$\Delta(c)\ge
\frac{\ell n}{(\alpha-\varepsilon)k(\ell-1+t)}$.
On the other hand, there are currently $k+t$ candidates, and the sum of their marginal losses is at most $n$. Since Reverse Sequential PAV deletes a candidate of minimum marginal loss,
$\Delta(c)\le\frac{n}{k+t}$.

Combining the two inequalities gives
$\frac{\ell n}{(\alpha-\varepsilon)k(\ell-1+t)}
\le \frac{n}{k+t}$,
and hence
$\frac{1}{\alpha-\varepsilon}
\le \frac{k(\ell-1+t)}{\ell(k+t)}$.
We note here that for an $\ell$-cohesive group we would have $1\leq \frac{k(\ell-1+t)}{\ell(k+t)} \Leftrightarrow 
\ell t \leq kt-k 
\Leftrightarrow 
\ell \leq k(1-\frac1t)$.
It holds that
$\frac{k(\ell-1+t)}{\ell(k+t)}
\le
\frac{kt}{k+t}$,
because
$\frac{\ell-1+t}{\ell}\le t$
is equivalent to
$(\ell-1)(t-1)\ge0$.
Finally, since $t\le r$ and the function
$t\mapsto kt/(k+t)$ is increasing,
$\frac{1}{\alpha-\varepsilon}
\le \frac{kt}{k+t}
\le \frac{kr}{k+r}
= \frac1\alpha$.
This contradicts $\varepsilon>0$.
\end{proof}

\begin{corollary}
\label{cor:approxejr}
    Let $W$ be a committee returned by Reverse Sequential PAV.
The committee $W$ satisfies $\frac{1}{r}$-EJR.
Moreover, even when $W$ fails EJR, it holds that for every integer $\ell\ge1$ satisfying
$\ell>k\left(1-\frac1r\right)$,
every $\ell$-cohesive group satisfies the EJR requirement.
\end{corollary}

\begin{proof}
If $r\leq2$, both claims follow from Theorem~1.
Suppose therefore that $r\geq3$.
Applying Theorem~3 with $\varepsilon=1/k$ gives
$(1/r)$-EJR, since
$\frac{k+r}{kr}-\frac1k=\frac1r$.
For the second claim, recall from the proof of Theorem~\ref{thm:revseq-alpha-ejr} that if an $\ell$-cohesive group violates the EJR requirement and $t\le r$ deletions remain when one of its commonly approved candidates is deleted, then
$\ell\le k\left(1-\frac1t\right)$.
Since $t\le r$,
$\ell\le k\left(1-\frac1t\right)\le k\left(1-\frac1r\right)$.
Hence no $\ell$-cohesive group with
$\ell>k\left(1-\frac1r\right)$
can violate the EJR requirement.
\end{proof}

We now turn to proportionality degree, for which the performance of RevSeqPAV was previously unknown \cite{lac-sko:multiwinner-book}. While this metric yields substantial guarantees for PAV and SeqPAV, our result shows that the same is not true for RevSeqPAV, at least on unrestricted instances. This is consistent with our earlier negative finding for EJR in the unrestricted setting (\Cref{thm:revseq-ejr-three-deletions}).

\begin{theorem}
\label{thm:revseq-prop-degree-zero}
For every integer $\ell\geq 1$, the proportionality degree of Reverse Sequential PAV at $\ell$ is exactly $0$.
\end{theorem}

\begin{proof}
    First we note that
$d_{\mathrm{RevSeqPAV}}(\ell)\ge 0$ so it suffices to show that
$d_{\mathrm{RevSeqPAV}}(\ell)\le 0$.
To this end, fix $\ell\ge1$ and let $d=4\ell$. We construct an election containing an $\ell$-cohesive group whose members obtain no approved winner.
Let
$C^\star=\{c_1,\ldots,c_\ell\}$
be a set of $\ell$ common candidates. Create $d$ pairwise disjoint blocks
$X_1,\ldots,X_d$,
each containing exactly $d$ candidates. Finally, create a set
$Y=\{y_1,\ldots,y_k\}$
of $k=d(d+1)$ candidates.
The voters are as follows where
$R=d(3d^2+7d+3)$.

\begin{center}
\begin{tabular}{c|c|c}
Voter type & Number of voters & Approved candidates\\
\hline
$V_j$, for each $j\in\{1,\ldots,d\}$ & $d+1$ & $C^\star\cup X_j$\\
$U_x$, for each $x\in\bigcup_jX_j$ & $d$ & ${x}$\\
$Z$ & $R$ & $Y$
\end{tabular}
\end{center}

Let $V=\bigcup_{j=1}^dV_j$. Then
$|V|=d(d+1)$.
There are $d^2$ candidates in the sets $X_j$, so the second type of voters corresponds to $d^3$ voters. Therefore the total number of voters is
$n=d(d+1)+d^3+R=4d(d+1)^2$.

Since $d=4\ell$ and $k=d(d+1)$, we have
$\ell n/k=(d/4)\cdot 4d(d+1)^2/[d(d+1)]=d(d+1)=|V|$.
Moreover, every voter in $V$ approves all $\ell$ candidates in $C^\star$. Hence $V$ is $\ell$-cohesive.
We will show that Reverse Sequential PAV deletes first all candidates in $C^\star$, then all candidates in the sets $X_j$, leaving precisely $Y$ as the final committee.\looseness-1

Suppose that $q\ge1$ candidates from $C^\star$ remain and no candidate from any $X_j$ or from $Y$ has yet been deleted. A remaining common candidate $c\in C^\star$ has marginal loss
$\Delta(c)=d(d+1)/(q+d)$.
For any $x\in X_j$, its $d+1$ voters of type $V_j$ each currently approve $q+d$ candidates, while its $d$ singleton voters contribute one each. Hence
$\Delta(x)=(d+1)/(q+d)+d$.
The difference is
$\Delta(x)-\Delta(c)=(dq+1)/(q+d)>0$.
Thus every remaining common candidate has strictly smaller marginal loss than every candidate in the sets $X_j$.
For every $y\in Y$,
$\Delta(y)=R/k=(3d^2+7d+3)/(d+1)>d$,
while
$\Delta(c)\le d$.
Thus every common candidate also has strictly smaller marginal loss than every candidate in $Y$.
Consequently, Reverse Sequential PAV deletes all candidates in $C^\star$ before deleting anything else.

Now suppose all common candidates have been deleted. Consider a block $X_j$ with $s\ge1$ candidates still remaining. For every remaining $x\in X_j$,
$\Delta(x)=d+(d+1)/s\le2d+1$.
On the other hand, every $y\in Y$ still has
$\Delta(y)=(3d^2+7d+3)/(d+1)>2d+1$.
Therefore, as long as any candidate from some $X_j$ remains, Reverse Sequential PAV deletes an $X$-candidate rather than a $Y$-candidate.
It follows that all candidates in $C^\star\cup X_1\cup\cdots\cup X_d$ are deleted before any candidate in $Y$. Since $|Y|=k$, the final committee is exactly
$W=Y$.
\end{proof}

We note that RevSeqPAV can equivalently be viewed as a ranking rule. Starting from the full candidate set, the candidate deleted first is placed last in the ranking, the candidate deleted second is placed second-to-last, and so on. Consequently, for every $k$, the top-$k$ prefix of the resulting ranking is exactly the committee returned by RevSeqPAV when the target committee size is $k$.

We now recall a proportionality notion for rankings introduced by \citet{proprank}. 
The setting examined in that paper is as follows: each voter submits an approval ballot, yet the outcome of an election is a ranking over all candidates. For some $k$ we use $R_{\le k}$ to denote the first $k$ candidates in the output ranking.
Using the same construction as in \Cref{thm:revseq-prop-degree-zero}, we resolve a question left open in that work.
For a group $V$ and a prefix of length $k$, recall that its justifiable demand is
$jd(V,k)=\min\{\lfloor |V|k/n\rfloor,\lambda(V)\}$,
where $\lambda(V)=|\bigcap_{i\in V}A_i|$ is the number of candidates commonly approved by the group. Its average satisfaction in the prefix $R_{\le k}$ is
$\operatorname{avg}_V(R_{\le k})=\frac{1}{|V|}\sum_{i\in V}|A_i\cap R_{\le k}|$.
The proportionality quality of a ranking is the minimum ratio between the average satisfaction of a group $V$ in $R_{\le k}$ and the justifiable demand of $V$
over all values of $k$ and sets of voters $V$ such that $jd(V,k)>0$.\looseness-1

\begin{corollary}
\label{corol:qual}
    The worst-case proportionality quality of rankings returned by Reverse Sequential PAV is $0$.
\end{corollary}

\begin{proof}
Since proportionality quality is nonnegative, it suffices to construct a Reverse Sequential PAV ranking containing a prefix for which some group has average satisfaction $0$ while having positive justifiable demand.
Consider the construction used in Theorem~\ref{thm:revseq-prop-degree-zero} for some $\ell\ge1$. There, the final committee $W$ has size $k$, and there exists a group $V$ such that
$|V|=\ell n/k$,
$\lambda(V)\ge\ell$,
and
$|A_i\cap W|=0$
for every voter $i\in V$.
Let $R$ be the complete ranking produced by Reverse Sequential PAV. Since its top-$k$ prefix is precisely the set $W$, we have
$\operatorname{avg}_V(R_{\le k})=0$.
On the other hand,
$jd(V,k)=\min\{\lfloor |V|k/n\rfloor,\lambda(V)\}=\min\{\ell,\lambda(V)\}=\ell$.
Hence the worst-case proportionality quality of Reverse Sequential PAV is $0$.
\end{proof}

\section{Approximation Guarantees of RevSeqPAV}
\label{sec:approx}
In this section, we study how well RevSeqPAV approximates the maximum PAV score, i.e., how close the PAV score it achieves is to the optimum. SeqPAV guarantees a constant fraction of the optimal PAV score. As our first result shows, the same is far from true for RevSeqPAV on general instances.
Our proof leverages the observation by \citet{lac-sko:multiwinner-book} that for $k=1$ the rule's outcome can be counterintuitive.

\begin{theorem}
\label{thm:approxzero}
The worst-case approximation ratio of Reverse Sequential PAV with respect to the PAV score is $0$.\looseness-1
\end{theorem}

\begin{proof}
Fix an integer $d\ge2$ and say that $k=1$. Consider one candidate $a$ and $d$ pairwise disjoint sets
$X_1,\ldots,X_d$,
each containing exactly $d$ candidates. Thus there are $1+d^2$ candidates in total. 
The approval profile is as follows.

\begin{center}
\begin{tabular}{c|c|c}
Voter type & Number of voters & Approved candidates\\
\hline
$V_j$, for each $j\in\{1,\ldots,d\}$ & $d+1$ & $\{a\}\cup X_j$\\
$U_x$, for each $x\in\bigcup_jX_j$ & $d$ & ${x}$
\end{tabular}
\end{center}

Initially, candidate $a$ is approved by all $d(d+1)$ voters of the first type. Each such voter currently approves exactly $d+1$ candidates. Hence
$\Delta_C(a)=d(d+1)/(d+1)=d$.
Consider instead some candidate $x\in X_j$. It is approved by the $d+1$ voters in $V_j$, each of whom contributes $1/(d+1)$, and by its $d$ singleton voters, each of whom contributes $1$. Therefore
$\Delta_C(x)=(d+1)/(d+1)+d=d+1$.
Thus
$\Delta_C(a)=d<d+1=\Delta_C(x)$
for every $x\in\bigcup_jX_j$, and Reverse Sequential PAV deletes $a$ first.

Since the target committee size is $1$, the candidate eventually returned by Reverse Sequential PAV must therefore be some candidate $x\in\bigcup_jX_j$. Such a candidate is approved by exactly
$(d+1)+d=2d+1$
voters. Hence
$\operatorname{PAV}(W_{\mathrm{Rev}})=2d+1$.
By contrast, the singleton committee $\{a\}$ has PAV score equal to the approval score of $a$, namely
$\operatorname{PAV}(\{a\})=d(d+1)$.
For $d\ge2$, this is larger than $2d+1$, so $\{a\}$ is a PAV-optimal singleton committee. Consequently,
$\frac{\operatorname{PAV}(W_{\mathrm{Rev}})}{\operatorname{OPT}_1}=\frac{2d+1}{d(d+1)}$.
As $d$ tends to infinity, this ratio tends to $0$. Hence the worst-case approximation ratio of Reverse Sequential PAV with respect to the PAV score is $0$, even when the committee size is fixed to $k=1$. 
\end{proof}

\Cref{thm:approxzero} concerns the unrestricted worst-case. In the remainder of the paper, we identify restrictions under which RevSeqPAV admits positive approximation guarantees. 
Hence, as was the case with proportionality guarantees, there are realms where the rule performs well.
First, we identify $r$ as a crucial parameter again, as in \Cref{sec:prop}. 
We show that when only a small number of candidates are to be deleted, Reverse Sequential PAV achieves a strong approximation guarantee.
In particular, this guarantee is strong when $k$ constitutes a large fraction of $m$, and for every fixed number of deletions $r$ it converges to $1$ as $k$ tends to infinity.
Moreover, this result implies the existence of a regime, when $r-1<\frac{k}{e-1}\approx0.582k$, in which our lower bound for RevSeqPAV exceeds the standard \(1-1/e\) approximation guarantee known for SeqPAV.

\begin{theorem}
\label{thm:kmapprox}
Let $W$ be a commitee returned by Reverse Sequential PAV. Then
$\operatorname{PAV}(W)\ge \frac{k}{m-1}\operatorname{OPT}_k = \frac{k}{k+r-1}\operatorname{OPT}_k$.
\end{theorem}

\begin{proof}
For every $s\in\{k,\ldots,m\}$, let $C_s$ denote the set of $s$ candidates remaining after $m-s$ deletions of Reverse Sequential PAV. Thus $C_m=C$ is the full candidate set and $C_k=W$ is the final committee.

We first consider the first deletion. 
Recall that the first deletion produces a PAV-optimal set among all sets of size $m-1$.
Now let $W^\star$ be a PAV-optimal committee of size $k$, so that
$\operatorname{PAV}(W^\star)=\operatorname{OPT}_k$.
Since $m-1\ge k$, we can extend $W^\star$ to some set $S$ of size $m-1$ by adding arbitrary candidates. The PAV score is monotone with respect to set inclusion: adding candidates cannot decrease the number of approved selected candidates of any voter, and hence cannot decrease their harmonic utility. Therefore,
$\operatorname{PAV}(S)\ge\operatorname{PAV}(W^\star)=\operatorname{OPT}_k$.
Since $C_{m-1}$ maximizes the PAV score over all sets of size $m-1$, we conclude that
$\operatorname{PAV}(C_{m-1})\ge\operatorname{PAV}(S)\ge\operatorname{OPT}_k$.

We next bound the loss caused by each subsequent deletion.
Consider a step at which the current candidate set is $C_s$, where $k+1\le s\le m-1$. Recall that
$\Delta_{C_s}(c)=\operatorname{PAV}(C_s)-\operatorname{PAV}(C_s\setminus\{c\})$
is the marginal PAV loss caused by deleting $c$.
For a voter $i$, let $t_i=|A_i\cap C_s|$. If $t_i=0$, then voter $i$ contributes $0$ to the marginal loss of every candidate. If $t_i>0$, then deleting any one of the $t_i$ candidates in $A_i\cap C_s$ contributes
$H(t_i)-H(t_i-1)=1/t_i$
to the marginal loss of that candidate. Since there are $t_i$ such candidates, voter $i$ contributes in total
$t_i\cdot(1/t_i)=1$
to the sum of the marginal losses of all candidates in $C_s$.
Consequently,
$\sum_{c\in C_s}\Delta_{C_s}(c)=|\{i\in N:A_i\cap C_s\neq\emptyset\}|$.
On the other hand, every voter in $\{i\in N:A_i\cap C_s\neq\emptyset\}$ contributes at least $1$ to $\operatorname{PAV}(C_s)$. Therefore,
$\sum_{c\in C_s}\Delta_{C_s}(c)\le\operatorname{PAV}(C_s)$.
There are $s$ candidates in $C_s$. Therefore, there exists some candidate $c\in C_s$ satisfying
$\Delta_{C_s}(c)\le\operatorname{PAV}(C_s)/s$.
Therefore, if $c_s$ denotes the candidate deleted at this step, then
$\Delta_{C_s}(c_s)\le\operatorname{PAV}(C_s)/s$.
It follows that
$\operatorname{PAV}(C_{s-1})
=\operatorname{PAV}(C_s)-\Delta_{C_s}(c_s)
\ge\operatorname{PAV}(C_s)-\frac{\operatorname{PAV}(C_s)}{s}
=\frac{s-1}{s}\operatorname{PAV}(C_s)$.\looseness-1

Thus, when $s$ candidates remain, Reverse Sequential PAV preserves at least an $(s-1)/s$ fraction of the current PAV score.
We now apply this inequality successively, starting from $C_{m-1}$ and following the sequence of deletions down to $C_k=W$. 
For the first step,
$\operatorname{PAV}(C_{m-2})\ge \frac{m-2}{m-1}\operatorname{PAV}(C_{m-1})$.
For the next step,
$\operatorname{PAV}(C_{m-3})\ge \frac{m-3}{m-2}\operatorname{PAV}(C_{m-2})$.
Substituting the first inequality into the second gives
$\operatorname{PAV}(C_{m-3})\ge \frac{m-3}{m-2}\frac{m-2}{m-1}\operatorname{PAV}(C_{m-1})$.
Continuing in the same way, each subsequent deletion contributes an additional factor $(s-1)/s$. Since the final set is $C_k=W$, we obtain
$\operatorname{PAV}(W)\ge
\frac{k}{k+1}\cdot
\frac{k+1}{k+2}\cdots
\frac{m-3}{m-2}\cdot
\frac{m-2}{m-1}
\operatorname{PAV}(C_{m-1})$.
Equivalently,
$\operatorname{PAV}(W)\ge
\left(\prod_{s=k+1}^{m-1}\frac{s-1}{s}\right)
\operatorname{PAV}(C_{m-1})\ge
\frac{k}{m-1}\operatorname{PAV}(C_{m-1})
\ge
\frac{k}{m-1}\operatorname{OPT}_k$.
\end{proof}

We now present a second positive result, which applies to instances with a bounded number of approvals per voter and shows that RevSeqPAV can perform well even independently of the value of $r$.

\begin{theorem}
\label{thm:boundedballot}
Let $W$ be a committee returned by Reverse Sequential PAV. Then
$\operatorname{PAV}(W)\ge \frac{1}{b}\operatorname{OPT}_k$.
\end{theorem}

\begin{proof}
Let $W^\star$ be a PAV-optimal committee of size $k$.
Also, let
$I=W\cap W^\star$,
$D=W^\star\setminus W$,
and
$R=W\setminus W^\star$.
Thus $D$ contains the candidates that belong to the optimal committee but were deleted by Reverse Sequential PAV, while $R$ contains the candidates that survive in $W$ but do not belong to the optimal committee.
Since $|W|=|W^\star|$, we have $|D|=|R|$. Therefore we can pair every candidate $c\in D$ with a distinct candidate $\phi(c)\in R$ through a bijection $\phi:D\rightarrow R$.

We first compare the marginal loss of a deleted optimal candidate $c$ with the marginal loss of its paired surviving candidate $\phi(c)$.
Fix some $c\in D$, and let $S_c$ be the set of candidates that are still present immediately before Reverse Sequential PAV deletes $c$. Define
$\delta_c=\Delta_{S_c}(c)$.
Since $\phi(c)\in W$, candidate $\phi(c)$ survives until the end and is therefore still present in $S_c$. Reverse Sequential PAV deletes a candidate of minimum marginal loss, so
$\delta_c=\Delta_{S_c}(c)\le \Delta_{S_c}(\phi(c))$.
Moreover, the marginal loss of a fixed surviving candidate can only increase as other candidates are deleted. Indeed, if a voter approves $\phi(c)$, their contribution to its marginal loss is the reciprocal of the number of their currently remaining approved candidates. As candidates are deleted, this denominator can only decrease. Hence
$\Delta_{S_c}(\phi(c))\le \Delta_W(\phi(c))$.
Combining these inequalities gives
$\delta_c\le \Delta_W(\phi(c))$.

We next relate $\delta_c$ to the benefit that the optimal committee obtains from candidate $c$.
Consider adding $c$ to the common part $I$. Let voter $i$ approve $c$, and suppose that immediately before $c$ was deleted she approved $t_i$ candidates that were still present. Their contribution to $\delta_c$ was therefore $1/t_i$.
Since every ballot contains at most $b$ candidates, we have $t_i\le b$. Hence
$1\le b/t_i$.
When we add $c$ to $I$, voter $i$ gains at most $1$: if she currently approves $q_i$ candidates in $I$, their gain is $1/(q_i+1)\le1$. Therefore
$1/(q_i+1)\le b/t_i$.
Thus, voter by voter, the contribution to the marginal gain of adding $c$ to $I$ is at most $b$ times the contribution to the marginal loss of $c$ when Reverse Sequential PAV deleted it. Summing over all voters who approve $c$, we obtain
$\operatorname{PAV}(I\cup\{c\})-\operatorname{PAV}(I)\le b\delta_c$.

Note that the PAV score is monotone and submodular. Indeed, if a voter currently approves \(q\) selected candidates, the marginal contribution of an additional approved candidate is \(1/(q+1)\), which can only decrease as the selected set grows.
Now consider adding all candidates in $D$ to $I$. 
Since the PAV score is submodular, the total gain from adding all of them is at most the sum of their individual gains when each is added directly to $I$. Hence
$\operatorname{PAV}(W^\star)-\operatorname{PAV}(I)\le \sum_{c\in D}\bigl(\operatorname{PAV}(I\cup\{c\})-\operatorname{PAV}(I)\bigr)\le b\sum_{c\in D}\delta_c$.
Using the pairing between $D$ and $R$, together with $\delta_c\le\Delta_W(\phi(c))$, we get
$\sum_{c\in D}\delta_c\le\sum_{w\in R}\Delta_W(w)$.

It remains to bound the latter quantity. Starting from $W$, delete the candidates in $R$ one by one until only $I$ remains. The marginal loss of a candidate can only increase as other candidates are deleted. Therefore, for every $w\in R$, its marginal loss measured initially at $W$ is no larger than its marginal loss at the moment when it is actually removed in this sequence. Consequently,
$\sum_{w\in R}\Delta_W(w)\le \operatorname{PAV}(W)-\operatorname{PAV}(I)$.
Hence,
$\operatorname{PAV}(W^\star)-\operatorname{PAV}(I)\le b\bigl(\operatorname{PAV}(W)-\operatorname{PAV}(I)\bigr)
\Leftrightarrow \operatorname{PAV}(W^\star)\le b\operatorname{PAV}(W)-(b-1)\operatorname{PAV}(I)\le b\operatorname{PAV}(W)$.
\end{proof}

We conclude by establishing tightness for the result of \Cref{thm:boundedballot}.

\begin{theorem}
\label{thm:boundedballottight}
 For every integer $b\ge2$ and every $\varepsilon>0$, there exists an election for which
$\operatorname{PAV}(W)<(\frac1b+\varepsilon)\operatorname{OPT}_1$,
where $W$ is the committee returned by Reverse Sequential PAV.
\end{theorem}

\begin{proof}
Fix $b\ge2$ and an integer $t> 0$. Consider an instance where $k=1$ and there is one candidate $a$ and $t$ pairwise disjoint sets
$X_1,\ldots,X_t$,
each containing exactly $b-1$ candidates.
Consider a set of voters $V_j$, for each $j\in\{1,\ldots,t\}$ that contains $b$ voters each approving $\{a\}\cup X_j$. Also say that there are $t$ voters in each set $U_x$, for each $x\in\bigcup_jX_j,$ each one approving $x$.
Note that the maximum ballot size is $b$.

We first consider candidate $a$. It is approved by all $tb$ voters in the groups $V_1,\ldots,V_t$. Each of these voters currently approves exactly $b$ candidates, and therefore contributes $1/b$ to the marginal loss of $a$. Thus
$\Delta_C(a)=tb\cdot\frac1b=t$.
Now consider any candidate $x\in X_j$. Candidate $x$ is approved by the $b$ voters in $V_j$, who together contribute
$b\cdot\frac1b=1$,
and by its $t$ supporters, who contribute $t$ in total. Hence
$\Delta_C(x)=t+1$.
Therefore,
$\Delta_C(a)=t<t+1=\Delta_C(x)$
for every $x\in\bigcup_jX_j$. Consequently, Reverse Sequential PAV deletes $a$ in its first step.\looseness-1

The candidate eventually returned by the rule must be some candidate $x\in\bigcup_jX_j$. Every such candidate has $b+t$ approvals: $b$ from the voters in its group $V_j$ and $t$ from its singleton supporters. Therefore,
$\operatorname{PAV}(W)=b+t$.
By contrast, candidate $a$ has $tb$ approvals. For all sufficiently large $t$, we have $tb>b+t$, and hence $\{a\}$ is a PAV-optimal singleton committee. Thus
$\operatorname{OPT}_1=tb$.
It follows that
$\frac{\operatorname{PAV}(W)}{\operatorname{OPT}_1}
=\frac{t+b}{tb}
=\frac1b+\frac1t$.
As $t$ tends to infinity, this ratio converges to $1/b$. Hence, for every $\varepsilon>0$, choosing $t>1/\varepsilon$ and $t>\frac{b}{b-1}$ gives
$\frac{\operatorname{PAV}(W)}{\operatorname{OPT}_1}<\frac1b+\varepsilon$.
\end{proof}

\section{Conclusions}
Our results reveal a sharp contrast between the behavior of RevSeqPAV in adversarial and in restricted instances.
In the worst case, it may fail EJR and produce a ranking of proportionality quality $0$, has proportionality degree $0$, and provides no positive constant approximation of the optimal PAV score. On the positive side, it satisfies EJR for $r\leq2$ (and this is tight), provides a $\frac{1}{r}$-approximation of EJR and offers the representation guarantee of EJR to sufficiently large cohesive groups of voters; it also achieves good approximations of the optimal PAV-score when $k$ is close to $m$ (specifically $\frac{k}{m-1}$) or voters approve a few candidates (specifically $\frac{1}{b}$—and this is tight).

An interesting question is whether bounded ballot size also improves proportionality guarantees. It is also important to understand which other restrictions can yield a positive proportionality degree. More broadly, understanding the domains on which SeqPAV and RevSeqPAV coincide, or even return a PAV-optimal committee, is an important direction. Finally, further empirical comparisons, e.g. complementing those already conducted for SeqPAV, could clarify when the two rules behave similarly in practice, when their differences emerge, and how their performance compares across different families of elections and various metrics.\looseness-1

\small{
\paragraph{Acknowledgements.} 
G. Papasotiropoulos was supported by the European Union (ERC, PRO-DEMOCRATIC, 101076570). Views and opinions expressed are however those of the author only and do not necessarily reflect those of the European Union or the European Research Council. Neither the European Union nor the granting authority can be held responsible for them. ChatGPT-5.6 Sol assisted with developing counterexamples, refining proofs, and editing the manuscript; it also identified the connection to the proportional rankings literature and provided the argument establishing \Cref{corol:qual}. 
}

\vspace{-0.3cm}
\begin{figure}[h]
\includegraphics[width=0.2\linewidth]{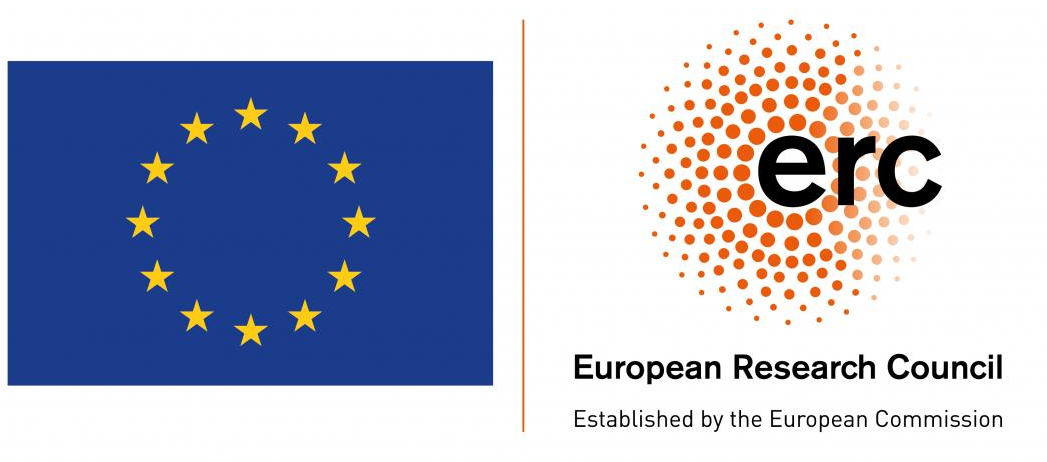}
\end{figure}

\bibliographystyle{apalike}
\bibliography{revseqpavbib}
\end{document}